\documentclass[letterpaper, 10 pt, conference]{ieeeconf}
\IEEEoverridecommandlockouts 
\usepackage{amsthm}
\usepackage{graphicx}
\usepackage[caption=false,font=footnotesize]{subfig}
\usepackage{amsmath,amssymb,amsfonts,mathtools}
\usepackage{tikz}
\usetikzlibrary{arrows.meta,positioning,calc,arrows}
\usepackage{pgfplots}
\usepgfplotslibrary{fillbetween}
\pgfplotsset{
  direction field/.style={
    rcGray, line width=0.35pt,
    quiver={u=\thisrow{u}, v=\thisrow{v}, scale arrows=1,
            every arrow/.append style={-{Stealth[length=1.8pt, width=1.5pt]}}},
  },
}
\pgfplotsset{compat=1.18,
  every axis/.append style={font=\scriptsize, line width=0.6pt, tick style={line width=0.4pt}, grid=none},
  legend style={draw=none, fill=none, font=\scriptsize, legend cell align=left},
  every axis label/.append style={font=\scriptsize},
  every axis x label/.append style={yshift=2pt},
  every axis y label/.append style={yshift=-4pt}}
\usepackage[dvipsnames]{xcolor}
\definecolor{rcNavy}{RGB}{20,35,80}
\definecolor{rcBlue}{RGB}{37,99,235}
\definecolor{rcPink}{RGB}{219,39,119}
\definecolor{rcGray}{RGB}{160,160,160}
\definecolor{rcTeal}{RGB}{13,148,136}

\usepackage{xcolor}
\usepackage[dvipsnames]{xcolor}
\usepackage[noabbrev]{cleveref}

\usepackage{enumerate}
\usepackage{mathtools}

\usepackage{cite}
\theoremstyle{definition}
\newtheorem{definition}{Definition}
\newtheorem{remark}{Remark}
\theoremstyle{plain}
\newtheorem{theorem}{Theorem}
\newtheorem{lemma}[theorem]{Lemma}
\newtheorem{proposition}[theorem]{Proposition}
\newtheorem{corollary}[theorem]{Corollary}

\usepackage[acronym,nomain,nowarn,nolist]{glossaries}
\newacronym{kkl}{KKL}{Kazantzis--Kravaris/Luenberger}
\newacronym{pbh}{PBH}{Popov--Belevitch--Hautus}
\newacronym{rms}{RMS}{root-mean-square}

\newcommand{\R}{\mathbb{R}}

\newcommand{\Rp}{\mathbb{R}_{+}}

\newcommand{\Wout}{W_{\mathrm{out}}}

\newcommand{\range}{\operatorname{range}}
\newcommand{\rank}{\operatorname{rank}}

\renewcommand{\Re}{\operatorname{Re}}

\newcommand{\dd}{\,\mathrm{d}}

\newcommand{\eqrefrange}[2]{\eqref{#1}--\eqref{#2}}

\newcounter{assumptioncounter}
\newenvironment{assumptionlist}
  {\begin{list}{}{\setlength{\leftmargin}{2.3em}\setlength{\labelwidth}{1.9em}
   \setlength{\labelsep}{0.4em}\setlength{\itemsep}{1pt}\setlength{\topsep}{3pt}
   \setlength{\parsep}{0pt}\setlength{\partopsep}{0pt}}}
  {\end{list}}
\newcommand{\assumption}[1]{\refstepcounter{assumptioncounter}
  \item[\textup{(A\arabic{assumptioncounter})}]\label{#1}}
\newcommand{\A}[1]{\textup{(A\ref{#1})}}

\usepackage[acronym,nomain,nowarn,nolist]{glossaries}
\newacronym{ESP}{ESP}{echo state property}
\newacronym{FBI}{FBI}{Francis--Byrnes--Isidori}
\newacronym{KKL}{KKL}{Kazantzis--Kravaris/Luenberger}

\title{\LARGE \bf The internal model principle in reservoir computing 
}

\author{Danyal Saqib and Alberto Padoan
\thanks{ D. Saqib and A. Padoan   are with the Department of Electrical and Computer Engineering, University of British Columbia, Vancouver, BC, V6T 1Z4, Canada (email: {\tt\small dsaqib@student.ubc.ca, alberto.padoan@ubc.ca}).  This work was supported by the Natural Sciences and Engineering Research Council of Canada (NSERC) (grant numbers: RGPIN-2025-06895 and DGECR-2025-00382).}
}

\begin{document}

\maketitle
\thispagestyle{empty}
\pagestyle{empty}

\begin{abstract}
The paper studies learning in reservoir computing systems through the lens of the internal model principle. The key idea is to model training signals as outputs of an exogenous signal generator, thus connecting reservoir computing to output regulation, observer theory, and frequency-domain interpolation. For linear signal generators, we characterize the steady-state response of a linear reservoir through invariant subspaces, establishing connections with Luenberger observers and classical frequency-domain interpolation. For nonlinear signal generators, we obtain analogous results using invariant manifolds, connecting synchronization to Kazantzis--Kravaris/Luenberger observers and readout training to nonlinear moment matching.  Case studies illustrate how our perspective explains possible instability of trained predictors and informs reservoir design.  
\end{abstract}

\section{Introduction}

Reservoir computing is a computationally efficient framework for modeling temporal data \cite{jaeger_echo_2001,maass_real_2002}. The idea is to drive a fixed dynamical system, the reservoir, with an input signal to be predicted or whose patterns are to be recognized. Learning consists of fitting a weighted sum of the reservoir state variables, the readout, by linear regression \cite{jaeger_echo_2001,maass_real_2002}.  Since the readout is the only trained parameter, any dynamical system with suitable forgetting properties can act as a reservoir \cite{jaeger_echo_2001,maass_real_2002}, so neural networks \cite{jaeger_echo_2001,maass_real_2002}, physical systems \cite{tanaka_recent_2019,stepney_physical_2024}, and even biological substrates \cite{sumi_biological_2023} can serve as computational media. The requirement that a reservoir forget its initial condition is formalized by the \gls{ESP} \cite{jaeger_echo_2001,yildiz_re-visiting_2012,manjunath_echo_2013}, under which reservoir computing can approximate broad classes of input--output maps \cite{grigoryeva_echo_2018}. Reservoir selection and training nevertheless remain largely guided by empirical considerations \cite{lukosevicius_reservoir_2009,zhang_survey_2023}, often relying on large, nonlinear, randomly generated reservoirs. For instance, the \gls{ESP} guarantees neither the stability nor the robustness of the resulting predictors \cite{lu_attractor_2018,haluszczynski_good_2019}, motivating a principled approach to reservoir analysis and design.   

The paper shows that learning in reservoir computing is a manifestation of the internal model principle \cite{francis_internal_1976,wonham_linear_1985,isidori_nonlinear_1995}, according to which output regulation against a class of exogenous signals requires an internal model of the dynamics generating them. By modeling training signals as outputs of an exogenous signal generator, we show that, after a transient, the reservoir state is determined by that of the signal generator through a synchronization map \cite{lu_reservoir_2017,hart_generalised_2024}, whose graph is an invariant manifold of the interconnected system, and that feeding the trained readout back to the reservoir builds an internal model of the signal generator. This perspective connects reservoir computing with output regulation \cite{francis_internal_1976,wonham_linear_1985,isidori_nonlinear_1995}, observer design \cite{luenberger_introduction_1971,kazantzis_nonlinear_1997,andrieu_existence_2006}, and moment matching \cite{antoulas_approximation_2005,astolfi_model_2010}, opening new avenues for reservoir analysis and design using classical control-theoretic tools.

\emph{Contributions.} Our contributions may be summarized as follows.
For linear signal generators, we characterize the steady-state response of a linear reservoir through an invariant subspace determined by a Sylvester equation.  
This connects synchronization maps \cite{lu_reservoir_2017,hart_generalised_2024} to Luenberger observers \cite{luenberger_introduction_1971} and establishes the equivalence between  exact signal reproduction and classical frequency-domain interpolation \cite{antoulas_approximation_2005}.  For nonlinear signal generators, we characterize the steady-state response through an invariant manifold, thus connecting synchronization to \gls{kkl} observers \cite{kazantzis_nonlinear_1997,andrieu_existence_2006} and readout training to nonlinear moment matching \cite{astolfi_model_2010}.  Finally, simple case studies illustrate the theory and its implications for reservoir computing, showing why exact prediction with linear reservoirs requires finite-dimensional linear realizations, how unassigned closed-loop modes can destabilize prediction, and how information about the origin of training signals enables reservoir design with prescribed performance. 

\emph{Related work.} 
Theoretical studies of reservoir computing have largely focused on the \gls{ESP}, fading memory, and approximation properties \cite{jaeger_echo_2001,grigoryeva_echo_2018}, and on synchronization between the driving system and the reservoir \cite{lu_reservoir_2017,hart_generalised_2024}. In particular, \cite{hart_generalised_2024} establishes conditions for generalized synchronization in continuous-time reservoirs and gives, for linear reservoirs, an integral expression of the synchronization map, \cite{lu_attractor_2018} studies the role of transverse stability in autonomous prediction, and \cite{singh_echo_2025} offers a systems perspective on echo state networks. Closely related work developed independently and concurrently connects reservoir computing with \gls{kkl} observers \cite{bizyaeva_kkl_2026}, and the low-dimensional invariant manifolds of trained reservoirs are studied in \cite{kaplan_dominant_2026}. None of these works develops the connections with the internal model principle, output regulation \cite{francis_internal_1976,wonham_linear_1985,isidori_nonlinear_1995}, Luenberger observers \cite{luenberger_introduction_1971}, and moment matching \cite{antoulas_approximation_2005,astolfi_model_2010} pursued in this paper, which, for linear reservoirs, are made explicit through Sylvester and invariance equations. 
 
\begin{figure*}[!t]
    \centering
    \subfloat[Training stage ($t<T$)]{\resizebox{0.47\textwidth}{!}{\definecolor{rcNavy}{RGB}{20,35,80}
\definecolor{rcBlue}{RGB}{37,99,235}
\definecolor{rcPink}{RGB}{219,39,119}
\definecolor{rcGray}{RGB}{175,175,175}
\begin{tikzpicture}[scale=0.75, every node/.style={scale=0.75}, line cap=round, line join=round,
  trainpath/.style={draw=rcBlue},  predpath/.style={draw=rcGray},  sharedpath/.style={draw=black},
  trainfill/.style={fill=rcBlue}, predfill/.style={fill=rcGray},
  trainblock/.style={draw=black}, predblock/.style={draw=rcGray}, sharedblock/.style={draw=black},
  trainlabel/.style={text=black}, predlabel/.style={text=rcGray}, sharedlabel/.style={text=black},
  traintext/.style={text=rcBlue}, predtext/.style={text=rcGray},
  wavein/.style={draw=rcBlue, line width=0.9pt}, waveout/.style={draw=rcGray, line width=0.9pt},
  switcharrow/.style={-latex, draw=black}]

\draw[trainblock] (-6.6,-0.8) rectangle (-3.9,0.8);
\node[trainlabel] at (-5.25,0) {%
$\begin{array}{rcl}
\dot{w} \!\!&\!\! = \!\!&\!\! s(w) \\
v \!\!&\!\! = \!\!&\!\! \ell(w) \\
y \!\!&\!\! = \!\!&\!\! q(w)
\end{array}$};
\draw[sharedblock] (-0.9,-0.5) rectangle (1.3,0.5);
\node[sharedlabel] at (0.2,0) {$\dot{x} = f(x,u)$};
\draw[predblock] (2.4,-0.5) rectangle (3.6,0.5);
\node[predlabel] at (3.0,0) {$W_{\mathrm{out}}$};
\draw[trainblock] (-1.0,-3.2) rectangle (5.1,-1.8);
\node[trainlabel] at (2.05,-2.5) {$\displaystyle \min_{W_{\mathrm{out}} \in \mathbb{R}^{p \times n}} \int_0^T \big\| W_{\mathrm{out}}\,x(t) - y(t) \big\|^2 \, dt$};

\coordinate (P) at (-1.5,0);
\coordinate (A) at ($(P)+(180:0.9)$);
\coordinate (B) at ($(P)+(135:0.9)$);
\draw[trainpath] (A) circle (0.06);
\draw[predpath] (B) circle (0.06);
\fill (P) circle (0.06);
\draw[sharedpath] (P) -- ($(P)+(176:0.9)$);   

\coordinate (P2) at (-2.4,-2.5);
\coordinate (C) at ($(P2)+(180:0.9)$);
\draw[trainpath] (C) circle (0.06);
\fill[trainfill] (P2) circle (0.06);
\draw[trainpath] (P2) -- ($(P2)+(176:0.9)$);   

\draw[trainpath] (-3.9,0) -- ($(A)+(-0.06,0)$);
\node[traintext] at (-3.3,-0.3) {$v$};
\draw[sharedpath, -latex] (P) -- (-0.9,0);
\node at (-1.2,-0.3) {$u$};
\draw[sharedpath, -latex] (1.3,0) -- (2.4,0);
\node at (1.55,0.3) {$x$};
\fill[trainfill] (1.85,0) circle (0.06);
\draw[trainpath, -latex] (1.85,0) -- (1.85,-1.8);
\draw[predpath, -latex] (3.6,0) -- (5.7,0);
\node[predtext] at (5.35,-0.3) {$\hat{y}$};
\fill[predfill] (4.2,0) circle (0.06);
\draw[predpath] (4.2,0) -- (4.2,1.4) -- (B |- 0,1.4) -- ($(B)+(0,0.06)$);
\node[predtext] at (1.0,1.7) {$u = W_{\mathrm{out}}\,x$};
\draw[trainpath] (-5.25,-0.8) -- (-5.25,-2.5) -- ($(C)+(-0.06,0)$);
\node[traintext] at (-5.55,-1.65) {$y$};
\draw[trainpath, -latex] (P2) -- (-1.0,-2.5);
\draw[trainpath, -latex] (3.0,-1.8) -- (3.0,-0.5);

\begin{scope}[shift={(-3.85,0.45)}]
  \draw[wavein] plot[domain=0:0.9, samples=80, smooth] (\x, {0.13*sin(800*\x) + 0.06*sin(1400*\x)});
\end{scope}
\begin{scope}[shift={(4.5,0.45)}]
  \draw[waveout] plot[domain=0:0.9, samples=80, smooth] (\x, {0.13*sin(800*\x+40) + 0.06*sin(1400*\x+40)});
\end{scope}

\end{tikzpicture}}\label{fig:training_diagram}}
    \hfil
    \subfloat[Prediction stage ($t\geq T$)]{\resizebox{0.47\textwidth}{!}{\definecolor{rcNavy}{RGB}{20,35,80}
\definecolor{rcBlue}{RGB}{37,99,235}
\definecolor{rcPink}{RGB}{219,39,119}
\definecolor{rcGray}{RGB}{175,175,175}
\begin{tikzpicture}[scale=0.75, every node/.style={scale=0.75}, line cap=round, line join=round,
  trainpath/.style={draw=rcGray},  predpath/.style={draw=rcPink},  sharedpath/.style={draw=black},
  trainfill/.style={fill=rcGray}, predfill/.style={fill=rcPink},
  trainblock/.style={draw=rcGray}, predblock/.style={draw=black}, sharedblock/.style={draw=black},
  trainlabel/.style={text=rcGray}, predlabel/.style={text=black}, sharedlabel/.style={text=black},
  traintext/.style={text=rcGray}, predtext/.style={text=rcPink},
  wavein/.style={draw=rcGray, line width=0.9pt}, waveout/.style={draw=rcPink, line width=0.9pt},
  switcharrow/.style={-latex, draw=black}]

\draw[trainblock] (-6.6,-0.8) rectangle (-3.9,0.8);
\node[trainlabel] at (-5.25,0) {%
$\begin{array}{rcl}
\dot{w} \!\!&\!\! = \!\!&\!\! s(w) \\
v \!\!&\!\! = \!\!&\!\! \ell(w) \\
y \!\!&\!\! = \!\!&\!\! q(w)
\end{array}$};
\draw[sharedblock] (-0.9,-0.5) rectangle (1.3,0.5);
\node[sharedlabel] at (0.2,0) {$\dot{x} = f(x,u)$};
\draw[predblock] (2.4,-0.5) rectangle (3.6,0.5);
\node[predlabel] at (3.0,0) {$W_{\mathrm{out}}$};
\draw[trainblock] (-1.0,-3.2) rectangle (5.1,-1.8);
\node[trainlabel] at (2.05,-2.5) {$\displaystyle \min_{W_{\mathrm{out}} \in \mathbb{R}^{p \times n}} \int_0^T \big\| W_{\mathrm{out}}\,x(t) - y(t) \big\|^2 \, dt$};

\coordinate (P) at (-1.5,0);
\coordinate (A) at ($(P)+(180:0.9)$);
\coordinate (B) at ($(P)+(135:0.9)$);
\draw[trainpath] (A) circle (0.06);
\draw[predpath] (B) circle (0.06);
\fill (P) circle (0.06);
\draw[sharedpath] (P) -- ($(P)+(135:0.9)$);   

\coordinate (P2) at (-2.4,-2.5);
\coordinate (C) at ($(P2)+(180:0.9)$);
\draw[trainpath] (C) circle (0.06);
\fill[trainfill] (P2) circle (0.06);
\draw[trainpath] (P2) -- ($(P2)+(150:0.9)$);   

\draw[trainpath] (-3.9,0) -- ($(A)+(-0.06,0)$);
\node[traintext] at (-3.3,-0.3) {$v$};
\draw[sharedpath, -latex] (P) -- (-0.9,0);
\node at (-1.2,-0.3) {$u$};
\draw[sharedpath, -latex] (1.3,0) -- (2.4,0);
\node at (1.55,0.3) {$x$};
\fill[trainfill] (1.85,0) circle (0.06);
\draw[trainpath, -latex] (1.85,0) -- (1.85,-1.8);
\draw[predpath, -latex] (3.6,0) -- (5.7,0);
\node[predtext] at (5.35,-0.3) {$\hat{y}$};
\fill[predfill] (4.2,0) circle (0.06);
\draw[predpath] (4.2,0) -- (4.2,1.4) -- (B |- 0,1.4) -- ($(B)+(0,0.06)$);
\node[predtext] at (1.0,1.7) {$u = W_{\mathrm{out}}\,x$};
\draw[trainpath] (-5.25,-0.8) -- (-5.25,-2.5) -- ($(C)+(-0.06,0)$);
\node[traintext] at (-5.55,-1.65) {$y$};
\draw[trainpath, -latex] (P2) -- (-1.0,-2.5);
\draw[trainpath, -latex] (3.0,-1.8) -- (3.0,-0.5);

\begin{scope}[shift={(-3.85,0.45)}]
  \draw[wavein] plot[domain=0:0.9, samples=80, smooth] (\x, {0.13*sin(800*\x) + 0.06*sin(1400*\x)});
\end{scope}
\begin{scope}[shift={(4.5,0.45)}]
  \draw[waveout] plot[domain=0:0.9, samples=80, smooth] (\x, {0.13*sin(800*\x+40) + 0.06*sin(1400*\x+40)});
\end{scope}

\end{tikzpicture}}\label{fig:prediction_diagram}}
        \caption{Reservoir computing. Left: during training, the reservoir is driven by the signal $v$ produced by the exosystem, the signals $x$ and $y$ are recorded, and the readout $\Wout$ is obtained through a least squares fit. Right: for prediction, the loop is closed by feeding the readout back to the reservoir input.}
    \label{fig:reservoir_tikz}
\end{figure*}
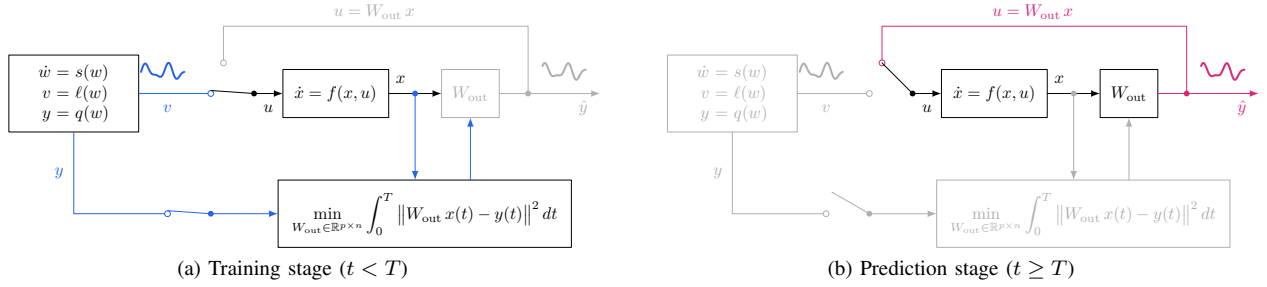

\emph{Paper organization.} The remainder of the paper is organized as follows. Section \ref{sec:rc} provides a brief overview of reservoir computing. Section \ref{sec:linear} studies linear reservoirs driven by linear signal generators. Section \ref{sec:nonlinear} extends these results to nonlinear reservoir computing systems. Section \ref{sec:design} illustrates the resulting implications for reservoir analysis and design. Section \ref{sec:conclusion} concludes the paper with a summary and an outlook to future research directions.

\section{Reservoir computing}\label{sec:rc}
Reservoir computing exploits the dynamics of a system, the reservoir, to generate features from a signal of interest, while learning only a linear readout of its state \cite{jaeger_echo_2001,maass_real_2002}, as illustrated in {Fig.\,\ref{fig:reservoir_tikz}}. 
The underlying premise is that sufficiently rich reservoir dynamics provide a representation of the input history from which the target can be approximated by a linear readout.  During prediction, the reservoir runs autonomously, with the readout fed back to its input to reproduce the signal.

The reservoir, in its most general form, is a dynamical system described by the equation \cite{singh_echo_2025}
\begin{align}
    \dot{x} &= f(x,u),
    \label{eq:general_reservoir}
\end{align}
where ${x(t)\in\R^n}$ is the reservoir state, ${u(t)\in\R^m}$ is the input, and ${f:\R^n\times\R^m\to\R^n}$ is locally Lipschitz. The reservoir output ${\hat{y}(t)\in\R^p}$ is a linear function of the state,
\begin{align}
    \hat{y} &= \Wout x , \quad \Wout\in\R^{p\times n} .
    \label{eq:linear_readout_definition}
\end{align}
Given a training horizon ${T>0}$, an input ${u(t)\in\R^m}$ and a target signal ${y(t)\in\R^p}$, the readout is typically obtained from the least-squares problem
\begin{align}
    \min_{\Wout\in\R^{p\times n}} \int_0^T \|\Wout x - y\|^2 \dd t ,
    \label{eq:least_squares_minimization}
\end{align}
where $\|y\|$ denotes the Euclidean norm of ${y\in\R^p}$.  Regularized variants, such as ridge regression, are also commonly used \cite{lukosevicius_reservoir_2009}. For prediction, the target is commonly chosen as the input itself, \textit{i.e.}, ${y=u}$. 

Throughout the paper, we consider the class $\mathcal U$ of piecewise continuous and bounded inputs ${u(t)\in\R^m}$, which we call \emph{admissible}. More general classes are studied in \cite{grigoryeva_echo_2018}. For ${x_0\in\R^n}$ and ${u\in\mathcal U}$, we denote by $\varphi_t(x_0,u)$ the solution of \eqref{eq:general_reservoir} at time $t$ from ${x(0)=x_0}$ under $u$. For an autonomous system, ${\dot x=f(x)}$, we abbreviate this notation to $\varphi_t(x_0)$. We also make the following standing assumption.
\begin{assumptionlist}
\item[\textup{(SA)}] The reservoir \eqref{eq:general_reservoir} is forward complete, i.e., $\varphi_t(x_0,u)$ is defined for all ${t\in\Rp}$, ${x_0\in\R^n}$, and ${u\in\mathcal U}$.
\end{assumptionlist}
\noindent
Assumption (SA) holds, \textit{e.g.}, if $f$ is globally Lipschitz in $x$ \cite{khalil_nonlinear_2002}, as for linear reservoirs and echo state networks \cite{jaeger_echo_2001}. 

A key requirement of reservoir computing is that the reservoir forget its initial condition when driven by any admissible input, as formalized by the \gls{ESP} \cite{jaeger_echo_2001,yildiz_re-visiting_2012}.
\begin{definition}[Echo state property]\label{def:esp}
The reservoir \eqref{eq:general_reservoir} has the \gls{ESP} if, for every ${u\in\mathcal U}$ and every ${x_0,x_0'\in\R^n}$, there holds
\begin{align}
    \lim_{t\to\infty}\|\varphi_t(x_0,u)-\varphi_t(x_0',u)\| = 0.
    \label{eq:esp}
\end{align}
\end{definition}
\noindent
The \gls{ESP} ensures that the effect of the initial condition vanishes asymptotically, so that, after a transient, the reservoir state is determined by the input history \cite{jaeger_echo_2001,yildiz_re-visiting_2012}. Under additional assumptions, reservoir computing systems satisfying the \gls{ESP} can approximate broad classes of input--output maps \cite{grigoryeva_echo_2018}. If the input is generated by a dynamical system with state $w$, then, under suitable conditions \cite{lu_reservoir_2017,hart_generalised_2024}, the reservoir state becomes asymptotically a function of $w$, that is, there exists a map $\pi$, called a \emph{synchronization map}, such that ${x(t)-\pi(w(t))\to0}$ as ${t\to\infty}$. For linear reservoirs, an integral expression of $\pi$ is given in \cite{hart_generalised_2024}, which Section \ref{sec:nonlinear} relates to \gls{kkl} observers. In the following sections, we specialize this general setup by assuming that the driving signal is generated by an exogenous dynamical system.

\section{Linear reservoir computing}\label{sec:linear}

This section studies linear reservoirs driven by linear signal generators, relating synchronization to Luenberger observers and output regulation, and readout training to frequency-domain interpolation.

\subsection{Problem setup}\label{sec:linear_setup}
In the linear setting, the general reservoir \eqref{eq:general_reservoir} reduces to
\begin{equation}
    \dot{x} = Fx +Gu,
    \label{eq:linear_reservoir}
\end{equation}
where ${x(t)\in\R^n}$, ${u(t)\in\R^m}$, ${F\in\R^{n\times n}}$, and ${G\in\R^{n\times m}}$. 
We will need the following assumptions for system \eqref{eq:linear_reservoir}.
\begin{assumptionlist}
\assumption{asmp:linear_hurwitz} The matrix $F$ is Hurwitz.
\assumption{asmp:linear_controllability} The pair $(F,G)$ is controllable.
\end{assumptionlist}
\noindent 
Assumption \A{asmp:linear_hurwitz} ensures that \eqref{eq:linear_reservoir} satisfies the \gls{ESP}. Indeed, for every input ${u\in\mathcal U}$, two solutions of \eqref{eq:linear_reservoir} driven by $u$ satisfy
\begin{align*}
    x(t)-x'(t) &= e^{Ft}\bigl(x(0)-x'(0)\bigr),
\end{align*}
which implies \eqref{eq:esp}. Assumption \A{asmp:linear_controllability} ensures that every reservoir mode is excited by the input, and is needed later to characterize the closed-loop modes left free by training.

We model the temporal data presented to the reservoir as signals generated by an autonomous dynamical system, referred to as the \emph{exosystem} or \emph{signal generator} in output regulation \cite{francis_internal_1976,isidori_nonlinear_1995}, observer design \cite{luenberger_introduction_1971}, and model reduction \cite{astolfi_model_2010}. We consider the linear exosystem
\begin{equation}
    \dot{w} = Sw, \quad v=Lw, \quad y=Qw,
    \label{eq:linear_exosystem}
\end{equation}
where ${w(t)\in\R^{\nu}}$ is the state, ${v(t)\in\R^m}$ is the generated signal, ${y(t)\in\R^p}$ is the target, and ${S\in\R^{\nu\times\nu}}$, ${L\in\R^{m\times\nu}}$, ${Q\in\R^{p\times\nu}}$. We make the following assumptions.
\begin{assumptionlist}
\assumption{asmp:linear_neutral} The matrix $S$ has distinct eigenvalues and $\sigma(S)\subset i\mathbb{R}$.
\assumption{asmp:linear_observability} The pair $(L,S)$ is observable.
\end{assumptionlist}
\noindent 
Assumption \A{asmp:linear_neutral}  implies that $S$ is diagonalizable over $\mathbb{C}$ and that every trajectory of \eqref{eq:linear_exosystem} is bounded and consists of a finite combination of constant and oscillatory modes. Consequently, the generated signals ${v=Lw}$ and ${y=Qw}$ are finite sums of constants and sinusoids. Assumption \A{asmp:linear_observability} excludes exosystem modes that are invisible in the generated signal $v$. In the prediction setting, the target coincides with the generated signal, that is, ${Q=L}$ and ${p=m}$.
Finally, we assume that the exosystem drives the reservoir through the interconnection defined by the equation
\begin{align}
    u &= v.
    \label{eq:interconnection}
\end{align}
\subsection{Connections to output regulation and observer design}\label{sec:linear_observer}
For the reservoir \eqref{eq:linear_reservoir} driven by the exosystem \eqref{eq:linear_exosystem} through the interconnection \eqref{eq:interconnection}, the goal of readout training is to find a readout ${W_{\mathrm{out}}\in\R^{p\times n}}$ minimizing the least-squares criterion \eqref{eq:least_squares_minimization}, whose minimum is zero if and only if some readout regulates to zero, on ${[0,T]}$, the readout error
\begin{align}
    e &= \hat{y}-y=W_{\mathrm{out}}x-Qw.
    \label{eq:linear_readout_error}
\end{align}
To connect readout training with output regulation and observer design, we first characterize the steady-state response of the interconnected system \eqrefrange{eq:linear_reservoir}{eq:interconnection}.
\begin{lemma}[Invariant subspace]\label{lem:linear_invariant_map}
Consider the reservoir \eqref{eq:linear_reservoir} driven by the exosystem \eqref{eq:linear_exosystem} through the interconnection \eqref{eq:interconnection}. Assume \A{asmp:linear_hurwitz} and \A{asmp:linear_neutral} hold. Then the Sylvester equation
\begin{align}
    F\Pi+GL &= \Pi S
    \label{eq:sylvester}
\end{align}
has a unique solution ${\Pi\in\mathbb{R}^{n\times\nu}}$, and the subspace
\begin{align}
    \mathcal{V} &= \left\{(w,x)\in\mathbb{R}^{\nu}\times\mathbb{R}^n \,\middle|\, x=\Pi w\right\}
    \label{eq:invariant_subspace}
\end{align}
is invariant under the dynamics of the system \eqrefrange{eq:linear_reservoir}{eq:interconnection}.
\end{lemma}
\begin{proof}
By \A{asmp:linear_hurwitz} and \A{asmp:linear_neutral}, ${\sigma(F)\cap\sigma(S)=\varnothing}$, hence \eqref{eq:sylvester} has a unique solution \cite{antoulas_approximation_2005}. Let $(w,x)$ be a solution of \eqrefrange{eq:linear_reservoir}{eq:interconnection} and let ${\varepsilon=x-\Pi w}$. Then \eqref{eq:sylvester} implies
\begin{align*}
    \dot\varepsilon = Fx+GLw-\Pi Sw = F\varepsilon+(F\Pi+GL-\Pi S)w = F\varepsilon .
\end{align*}
If ${(w(0),x(0))\in\mathcal V}$, then ${\varepsilon(0)=0}$, so ${\varepsilon(t)=0}$ for all ${t\geq0}$ by uniqueness of solutions of ${\dot\varepsilon=F\varepsilon}$, \textit{i.e.}, $\mathcal V$ is invariant.
\end{proof}
\noindent
Lemma \ref{lem:linear_invariant_map} shows that the interconnected system \eqrefrange{eq:linear_reservoir}{eq:interconnection} admits an invariant subspace on which the reservoir state is a linear function of the exosystem state. Moreover, \eqref{eq:sylvester} implies that $\range\Pi$ is a controlled invariant for the pair $(F,G)$ \cite{wonham_linear_1985}. The Sylvester equation \eqref{eq:sylvester} arises naturally in output regulation \cite{wonham_linear_1985} and the design of Luenberger observers \cite{luenberger_introduction_1971}. 

To analyze the training dynamics, we define the error
\begin{align}
    \varepsilon &= x-\Pi w .
    \label{eq:manifold_error_dynamics}
\end{align}
\begin{proposition}[Luenberger synchronization map]\label{prop:luenberger}
Consider the linear reservoir \eqref{eq:linear_reservoir} driven by the linear exosystem \eqref{eq:linear_exosystem} through the interconnection \eqref{eq:interconnection}. Assume \A{asmp:linear_hurwitz} and \A{asmp:linear_neutral} hold. Then the following statements hold.
{\renewcommand{\theenumi}{\roman{enumi}}\renewcommand{\labelenumi}{(\theenumi)}\setlength{\leftmargini}{1.75em}
\begin{enumerate}
    \item\label{item:pi_integral} The unique solution of \eqref{eq:sylvester} is
    \begin{align}
        \Pi &= \int_0^\infty e^{Ft}GL e^{-St}\dd t.
        \label{eq:pi_integral}
    \end{align}
    \item\label{item:sync_attractive} Along every solution of \eqrefrange{eq:linear_reservoir}{eq:interconnection}, the error \eqref{eq:manifold_error_dynamics} satisfies ${\dot{\varepsilon} = F\varepsilon}$, and $\mathcal{V}$ is globally exponentially attractive.
        \item\label{item:exact_readout} There exists ${\Wout\in\R^{p\times n}}$ such that
    \begin{align}
        \Wout\Pi &= Q
        \label{eq:perfect_readout}
    \end{align}
    if and only if ${\rank\begin{bsmallmatrix}\Pi\\ Q\end{bsmallmatrix}=\rank\Pi}$. Any such $\Wout$ gives a readout error \eqref{eq:linear_readout_error} which converges exponentially to zero.
\end{enumerate}
}
\end{proposition}
\begin{proof}
(\ref{item:pi_integral}) By \A{asmp:linear_hurwitz} and \A{asmp:linear_neutral}, the integral in \eqref{eq:pi_integral} is absolutely convergent. Denoting its value by $\tilde\Pi$, we obtain
\begin{align*}
    F\tilde\Pi-\tilde\Pi S
    &= \int_0^\infty \frac{\dd}{\dd t}\left(e^{Ft}GL e^{-St}\right)\dd t = -GL,
\end{align*}
so $\tilde\Pi$ solves \eqref{eq:sylvester} and ${\tilde\Pi=\Pi}$ by Lemma \ref{lem:linear_invariant_map}.

(\ref{item:sync_attractive}) The identity ${\dot\varepsilon=F\varepsilon}$ is established in the proof of Lemma \ref{lem:linear_invariant_map}. By \A{asmp:linear_hurwitz}, there exist ${M\geq1}$ and ${\alpha>0}$ such that ${\|e^{Ft}\|\leq Me^{-\alpha t}}$, so ${\|\varepsilon(t)\|\leq Me^{-\alpha t}\|\varepsilon(0)\|}$ for all ${t\geq0}$.

(\ref{item:exact_readout}) The rank condition holds if and only if the row space of $Q$ is contained in the row space of $\Pi$, \textit{i.e.}, if and only if ${Q=\Wout\Pi}$ for some ${\Wout\in\R^{p\times n}}$. For every $\Wout$ satisfying \eqref{eq:perfect_readout},
$    e=\Wout x-Qw=\Wout(x-\Pi w)=\Wout\varepsilon,$
which converges exponentially to zero by (\ref{item:sync_attractive}).
\end{proof}
\noindent
Proposition \ref{prop:luenberger} shows that, after a transient, the reservoir state is given by the synchronization map ${w\mapsto\Pi w}$ \cite{lu_reservoir_2017,hart_generalised_2024}, and that the reservoir may be interpreted as a Luenberger observer \cite{luenberger_introduction_1971} of $Qw$, with estimate $\hat y$, whenever \eqref{eq:perfect_readout} holds. We call \eqref{eq:perfect_readout} the \emph{exact readout condition}. Together with \eqref{eq:sylvester}, it defines the regulator equations \cite{francis_internal_1976,wonham_linear_1985} for driving the readout error \eqref{eq:linear_readout_error} to zero, with the input ${u=Lw}$ prescribed by the interconnection \eqref{eq:interconnection} and the readout $\Wout$ as the design variable. The condition  ${\rank\begin{bsmallmatrix}\Pi\\ Q\end{bsmallmatrix}=\rank\Pi}$ is equivalent to
${\ker\Pi\subseteq\ker Q}$, and holds if ${\rank\Pi=\nu}$, which is the case for ${m=1}$ and ${n\geq\nu}$ under assumptions \A{asmp:linear_hurwitz}--\A{asmp:linear_observability}  \cite{souza_controllability_1981}.

\subsection{Reservoir prediction}\label{subsec:forecasting}
Once training is complete, the reservoir may be used to predict the target signal by feeding its output back to its input \cite{jaeger_echo_2001}. In the prediction setting, with ${Q=L}$ and ${p=m}$, this defines the feedback law
\begin{align}
    u &= W_{\mathrm{out}}x,
    \label{eq:readout_feedback}
\end{align}
which gives the closed-loop dynamics
\begin{align}
    \dot{x} &= (F+GW_{\mathrm{out}})x.
    \label{eq:loop_closure}
\end{align}
\begin{proposition}[Closed-loop internal model]\label{prop:free_modes}
Consider the closed-loop reservoir \eqref{eq:loop_closure}, with ${p=m}$ and ${Q=L}$, where $\Wout$ satisfies the exact readout condition \eqref{eq:perfect_readout}. Assume \A{asmp:linear_hurwitz}, \A{asmp:linear_neutral}, and \A{asmp:linear_observability} hold. Then
\begin{align}
    (F+G\Wout)\Pi &= \Pi S,
    \label{eq:closed_loop_intertwining}
\end{align}
and $\Pi$ is injective, so that ${\sigma(S)\subseteq\sigma(F+G\Wout)}$.
\end{proposition}
\begin{proof}
Since ${Q=L}$, \eqref{eq:perfect_readout} and \eqref{eq:sylvester} give $${(F+G\Wout)\Pi=F\Pi+GL=\Pi S}.$$ If ${\Pi v=0}$, then ${Lv=\Wout\Pi v=0}$ and $${\Pi Sv=F\Pi v+GLv=0}$$ by \eqref{eq:sylvester}, so $\ker\Pi$ is an $S$-invariant subspace of $\ker L$, hence ${\ker\Pi=\{0\}}$ by \A{asmp:linear_observability}, that is, ${\rank\Pi=\nu}$. Thus, by \eqref{eq:closed_loop_intertwining}, the restriction of ${F+G\Wout}$ to $\range\Pi$ is similar to $S$, and ${\sigma(S)\subseteq\sigma(F+G\Wout)}$.
\end{proof}
\noindent
Proposition \ref{prop:free_modes} states that the subspace $\range\Pi$ is invariant under ${F+G\Wout}$, \textit{i.e.}, the readout feedback \eqref{eq:readout_feedback} is a \textit{friend} of $\range\Pi$ in the language of \cite{wonham_linear_1985}. Exact training thus embeds the exosystem in the closed loop, which reproduces every exosystem trajectory together with its target, and the modes on $\range\Pi$ form an internal model of the signal generator \cite{francis_internal_1976,wonham_linear_1985}. In particular, ${F+G\Wout}$ is not Hurwitz.

\begin{remark}[Implications for reservoir prediction]\label{rem:transverse}
Proposition \ref{prop:free_modes} implies that the closed-loop modes outside $\range\Pi$ are not fixed by training. To see this, let $\Pi^\dagger$ be the Moore--Penrose pseudoinverse of $\Pi$, and let $N$ have orthonormal columns spanning $\ker\Pi^\top$. In the coordinates ${(\Pi^\dagger x,N^\top x)}$, the matrix ${F+G\Wout}$ is block upper triangular, with diagonal blocks $S$ and ${N^\top(F+G\Wout)N}$, whose ${n-\nu}$ eigenvalues we call \emph{transverse}. Under \A{asmp:linear_controllability}, an application of the \gls{pbh} test \cite{wonham_linear_1985} shows that, by appropriately choosing $\Wout$ among the solutions of \eqref{eq:perfect_readout}, the transverse eigenvalues can be assigned to any multiset of ${n-\nu}$ complex numbers closed under conjugation \cite{wonham_linear_1985}. The stability of the predictor~\eqref{eq:linear_readout_definition}, \eqref{eq:loop_closure} is thus decided not by the target, but by the transverse eigenvalues, which depend on the transients and noise in the data entering \eqref{eq:least_squares_minimization}, and on regularization, if any, as observed for trained predictors in~\cite{lu_attractor_2018}.
\end{remark}
\subsection{A frequency-domain interpretation}\label{sec:linear_moments}
The Sylvester equation \eqref{eq:sylvester} also appears in model reduction by moment matching \cite{antoulas_approximation_2005,astolfi_model_2010}, where it characterizes the moments of a linear system at the eigenvalues of a signal generator. This connection allows us to interpret readout training as a frequency-domain interpolation problem.

Consider the linear reservoir \eqref{eq:linear_reservoir} with  output   \eqref{eq:linear_readout_definition}. Its transfer matrix from $u$ to $\hat{y}$ is
\begin{align}
    H(s) &= W_{\mathrm{out}}(sI_n-F)^{-1}G.
    \label{eq:transfer_function}
\end{align}
Following \cite{antoulas_approximation_2005,astolfi_model_2010}, the moment of order ${k\in\mathbb{N}}$ of the system at ${s^\star\in\mathbb{C}\setminus\sigma(F)}$ is
\begin{align*}
    \eta_k(s^\star) &= \frac{(-1)^k}{k!}\frac{\dd^kH}{\dd s^k}(s^\star),
\end{align*}
with the convention that ${\eta_0(s^\star)=H(s^\star)}$. Let $s_1,\ldots,s_\nu$ be the eigenvalues of $S$, with eigenvectors ${v_1,\ldots,v_\nu\in\mathbb{C}^{\nu}}$, and define ${V = [\,v_1\ \cdots\ v_\nu\,]}$. Under assumption \A{asmp:linear_neutral}, $V$ is invertible. Moreover, the \gls{pbh} observability test \cite{wonham_linear_1985} and \A{asmp:linear_observability} imply ${Lv_i\neq0}$ for ${i=1,\ldots,\nu}$. In the language of rational interpolation \cite{antoulas_approximation_2005}, each vector $Lv_i$ defines a nonzero tangential interpolation direction.
\begin{proposition}[Moment matching]\label{prop:moment_matching}
Consider the reservoir \eqref{eq:linear_reservoir} with output \eqref{eq:linear_readout_definition} and transfer matrix \eqref{eq:transfer_function}. Assume \A{asmp:linear_hurwitz} and \A{asmp:linear_neutral} hold. Then the following statements hold.
{\renewcommand{\labelenumi}{(\roman{enumi})}\setlength{\leftmargini}{2.5em}
\begin{enumerate}
    \item The unique solution $\Pi$ of \eqref{eq:sylvester} satisfies
    \begin{align}
        W_{\mathrm{out}}\Pi V
        =
        \begin{bmatrix}
            \eta_0(s_1)Lv_1 &
            \cdots &
            \eta_0(s_\nu)Lv_\nu
        \end{bmatrix}.
        \label{eq:tangential_moments}
    \end{align}
    \item The exact readout condition \eqref{eq:perfect_readout} holds if and only if
    \begin{align}
        \eta_0(s_i)Lv_i = Qv_i, \qquad i=1,\ldots,\nu.
        \label{eq:interpolation_conditions}
    \end{align}
\end{enumerate}
}
\end{proposition}
\begin{proof}
(i) Let $Sv_i=s_i v_i$. Right-multiplying \eqref{eq:sylvester} by $v_i$ gives ${(s_iI_n-F)\Pi v_i=GLv_i}$. By \A{asmp:linear_hurwitz} and \A{asmp:linear_neutral}, ${s_i\notin\sigma(F)}$, hence ${W_{\mathrm{out}}\Pi v_i=W_{\mathrm{out}}(s_iI_n-F)^{-1}GLv_i=\eta_0(s_i)Lv_i}$ for ${i=1,\ldots,\nu}$, which is \eqref{eq:tangential_moments}.

(ii) Since $V$ is invertible, \eqref{eq:perfect_readout} holds if and only if ${W_{\mathrm{out}}\Pi V=QV}$, \textit{i.e.}, by \eqref{eq:tangential_moments}, if and only if \eqref{eq:interpolation_conditions} holds.
\end{proof}
\noindent
Proposition \ref{prop:moment_matching} shows that exact readout training is equivalent to tangential interpolation of the reservoir transfer matrix at the eigenvalues of the exosystem, with interpolation directions $Lv_i$ and values $Qv_i$ \cite{antoulas_approximation_2005}. For single-input reservoirs, \textit{i.e.}, ${m=1}$, each $Lv_i$ is a nonzero scalar. Defining ${D = \operatorname{diag}(Lv_1,\ldots,Lv_\nu)}$ and ${\Gamma=VD^{-1}}$, both $D$ and $\Gamma$ are invertible, and \eqref{eq:tangential_moments} becomes
\begin{align*}
    W_{\mathrm{out}}\Pi\Gamma
    =
    \begin{bmatrix}
        \eta_0(s_1) & \cdots & \eta_0(s_\nu)
    \end{bmatrix}.
\end{align*}
Hence, for ${m=1}$, the moments of order zero of the trained reservoir \eqref{eq:linear_reservoir}, \eqref{eq:linear_readout_definition} at the eigenvalues of the exosystem \eqref{eq:linear_exosystem} are in one-to-one correspondence with the entries of ${W_{\mathrm{out}}\Pi}$ \cite[Lemma~1]{astolfi_model_2010}. Linear reservoir computing may thus be interpreted as enforcing moment matching in a data-driven manner, in line with \cite{scarciotti_data_2017,padoan_geometric_2017,padoan_least_2023}.

\section{Nonlinear reservoir computing}\label{sec:nonlinear}
This section extends the preceding results to nonlinear reservoirs and signal generators, using tools from output regulation \cite{isidori_nonlinear_1995} and  observer design \cite{kazantzis_nonlinear_1997,andrieu_existence_2006}. 
\subsection{Problem setup}
Consider the reservoir \eqref{eq:general_reservoir} and assume that the driving signal is generated by the nonlinear exosystem
\begin{equation}\label{eq:nonlinear_exosystem}
    \dot{w} = s(w), \qquad v = \ell(w), \qquad y = q(w),
\end{equation}
where ${w(t)\in\R^\nu}$ is the state, ${v(t)\in\R^m}$ is the generated signal, ${y(t)\in\R^p}$ is the target, ${s:\R^\nu\to\R^\nu}$, ${\ell:\R^\nu\to\R^m}$, and ${q:\R^\nu\to\R^p}$, and the reservoir is driven by the exosystem \eqref{eq:nonlinear_exosystem} through the interconnection \eqref{eq:interconnection}. As in the linear case, the goal of readout training is to find a readout ${W_{\mathrm{out}}\in\R^{p\times n}}$ minimizing \eqref{eq:least_squares_minimization}, whose minimum is zero if and only if some readout regulates to zero, on ${[0,T]}$, the readout error
\begin{equation}\label{eq:nonlinear_readout_error}
e = \hat{y} - y = \Wout x - q(w).
\end{equation}
We first characterize the steady-state response of the interconnected system \eqref{eq:general_reservoir}, \eqref{eq:interconnection}, \eqref{eq:nonlinear_exosystem}, which shows the existence of a synchronization map \cite{lu_reservoir_2017,hart_generalised_2024}. This requires assumptions on regularity, stability, and neutral stability that are standard in nonlinear output regulation \cite{isidori_nonlinear_1995}, and are formulated in terms of the linearizations of \eqref{eq:general_reservoir} and \eqref{eq:nonlinear_exosystem} at the origin, defined as
\begin{align}
    &F_0 = \frac{\partial f}{\partial x} (0,0),
    &&G_0 = \frac{\partial f}{\partial u} (0,0),\nonumber\\
    &S_0 = \frac{\partial s}{\partial w}(0),
    &&L_0 = \frac{\partial \ell}{\partial w}(0). \label{eq:nonlinear_linearization}
\end{align}
We will need the following assumptions.
\begin{assumptionlist}
\assumption{asmp:nonlinear_regularity} The map $f$ is of class $C^2$ on a neighborhood of $(0,0)$, with ${f(0,0)=0}$, and there exists an open neighborhood ${\mathcal O\subseteq\R^\nu}$ of ${w=0}$ on which $s$ and $\ell$ are of class $C^2$ and $q$ is continuous, with ${s(0)=0}$ and ${\ell(0)=0}$.
\assumption{asmp:nonlinear_stability} The matrix $F_0$ is Hurwitz.
\assumption{asmp:nonlinear_neutral_stability} The equilibrium ${w=0}$ of ${\dot w=s(w)}$ is stable in the sense of Lyapunov, every ${w\in\mathcal O}$ is positively Poisson stable\footnote{A point $w_0$ is positively Poisson stable for ${\dot w=s(w)}$ if ${\varphi_{t_k}(w_0)\to w_0}$ along some sequence ${t_k\to\infty}$ \cite{isidori_nonlinear_1995}.}, and ${\sigma(S_0) \subseteq i\mathbb R}$.
\end{assumptionlist}
\noindent
Assumption \A{asmp:nonlinear_regularity} makes the linearizations \eqref{eq:nonlinear_linearization} well defined and the origin an equilibrium of the interconnected system. Assumption \A{asmp:nonlinear_stability} is the nonlinear counterpart of \A{asmp:linear_hurwitz}. By a Lyapunov argument (see, \textit{e.g.}, \cite{khalil_nonlinear_2002}), any two solutions of \eqref{eq:general_reservoir} starting sufficiently close to the origin and driven by the same admissible input ${u\in\mathcal U}$ of sufficiently small amplitude converge to each other exponentially, \textit{i.e.}, the reservoir \eqref{eq:general_reservoir} has the \gls{ESP} locally. Assumption \A{asmp:nonlinear_neutral_stability} is the nonlinear counterpart of \A{asmp:linear_neutral}, and ensures that the exosystem \eqref{eq:nonlinear_exosystem} generates bounded signals which keep returning near their initial value, and that its linearization has neither stable nor unstable modes \cite{isidori_nonlinear_1995}.
\subsection{Connections to output regulation and observer design}\label{subsec:nonlinear_observer}
\begin{lemma}[Invariant manifold]\label{lem:nonlinear_synchronization}
Consider the reservoir \eqref{eq:general_reservoir} driven by the exosystem \eqref{eq:nonlinear_exosystem} through the interconnection \eqref{eq:interconnection}. Assume \A{asmp:nonlinear_regularity}--\A{asmp:nonlinear_neutral_stability} hold. Then the invariance equation
\begin{align}
    \frac{\partial\pi}{\partial w}(w)s(w)
    &=
    f(\pi(w), \ell(w))
    \label{eq:nonlinear_invariance}
\end{align}
has a locally unique solution ${\pi:\mathcal O_\pi\to\mathbb R^n}$ of class $C^2$ with ${\pi(0)=0}$, defined on an open neighborhood ${\mathcal O_\pi\subseteq\mathcal O}$ of ${w=0}$, and the manifold
\begin{align}
    \mathcal M
    &=
    \left\{(w,x)\in\mathcal O_\pi\times\mathbb R^n
    \,\middle|\,
    x=\pi(w)\right\}
    \label{eq:nonlinear_invariant_manifold}
\end{align}
is locally invariant under the dynamics of the interconnected system \eqref{eq:general_reservoir}, \eqref{eq:interconnection}, \eqref{eq:nonlinear_exosystem}.
\end{lemma}
\begin{proof}
The proof follows standard arguments of nonlinear output regulation (see, \textit{e.g.}, \cite{isidori_nonlinear_1995}). By \A{asmp:nonlinear_regularity}, the interconnected system is of class $C^2$ with an equilibrium at the origin. By \A{asmp:nonlinear_stability} and \A{asmp:nonlinear_neutral_stability}, its linearization at the origin has $n$ eigenvalues with negative real part and $\nu$ eigenvalues on the imaginary axis, so the center manifold theorem yields a locally invariant $C^2$ manifold which is the graph of a map $\pi$ with ${\pi(0)=0}$. Local invariance of this graph is equivalent to \eqref{eq:nonlinear_invariance}, and local uniqueness follows from the Lyapunov and Poisson stability in \A{asmp:nonlinear_neutral_stability}.
\end{proof}
\noindent
Lemma \ref{lem:nonlinear_synchronization} is the nonlinear counterpart of Lemma \ref{lem:linear_invariant_map}, where the invariance equation \eqref{eq:nonlinear_invariance} replaces the Sylvester equation \eqref{eq:sylvester}, and $\mathcal M$, a center manifold of the interconnected system, replaces $\mathcal V$. By the center manifold theorem \cite{isidori_nonlinear_1995}, $\mathcal M$ is locally exponentially attractive. After a transient, the reservoir state is thus determined by the exosystem state through the synchronization map $\pi$ \cite{lu_reservoir_2017,hart_generalised_2024}. Together with an exact readout condition, which will be stated in \eqref{eq:nonlinear_perfect_readout}, the invariance equation \eqref{eq:nonlinear_invariance} plays the role of the \gls{FBI} equations of nonlinear output regulation \cite{isidori_nonlinear_1995}, with the input ${u=\ell(w)}$ defined by the interconnection \eqref{eq:interconnection}.
Unlike in the linear case, these conclusions are local and $\pi$ does not admit an explicit expression, in general. For linear reservoirs, however, the next result gives an explicit expression of $\pi$ and shows that $\mathcal M$ attracts every reservoir state.
\subsubsection{Linear reservoirs as \gls{kkl} observers}\label{subsec:nonlinear_observer}
Consider the reservoir \eqref{eq:linear_reservoir} driven by the exosystem \eqref{eq:nonlinear_exosystem} through the interconnection \eqref{eq:interconnection}. The invariance equation \eqref{eq:nonlinear_invariance} then reads
\begin{align}
    \frac{\partial\pi}{\partial w}(w)s(w)
    &=
    F \pi(w)+ G \ell(w),
    \label{eq:kkl_invariance}
\end{align}
which is the equation underlying \gls{kkl} observers \cite{kazantzis_nonlinear_1997,andrieu_existence_2006}. The connection with \gls{kkl} observers  requires the following assumption on the set on which the exosystem \eqref{eq:nonlinear_exosystem} evolves.
\begin{assumptionlist}
\assumption{asmp:nonlinear_backward} The set ${\mathcal K\subseteq\mathcal O_\pi}$ is compact and invariant under the flow of ${\dot w=s(w)}$, that is, ${\varphi_t(\mathcal K)=\mathcal K}$ for all ${t\in\R}$.
\end{assumptionlist}
\noindent
Assumption \A{asmp:nonlinear_backward} ensures that the solutions of \eqref{eq:nonlinear_exosystem} starting in $\mathcal K$ are defined for all past and future times and remain in the domain of $\pi$. As in the linear case, define the error
\begin{align}
    \varepsilon &= x - \pi (w).
    \label{eq:nonlinear_manifold_error_dynamics}
\end{align}
\begin{proposition}[KKL synchronization map]\label{prop:kkl_synchronization}
Consider the reservoir \eqref{eq:linear_reservoir} driven by the exosystem \eqref{eq:nonlinear_exosystem} through the interconnection \eqref{eq:interconnection}. Assume \A{asmp:linear_hurwitz}, \A{asmp:nonlinear_regularity}, \A{asmp:nonlinear_neutral_stability}, and \A{asmp:nonlinear_backward} hold. Then the following statements hold.
{\renewcommand{\theenumi}{\roman{enumi}}\renewcommand{\labelenumi}{(\theenumi)}\setlength{\leftmargini}{2.em}
\begin{enumerate}
    \item\label{item:kkl_integral} The locally unique solution $\pi$ of \eqref{eq:kkl_invariance} is
    \begin{align}
        \pi(w) &= \int_0^\infty e^{Ft}G\ell\bigl(\varphi_{-t}(w)\bigr)\dd t, \quad w\in\mathcal K.
        \label{eq:nonlinear_synchronization_map}
    \end{align}
    \item\label{item:kkl_sync} Along every solution of \eqref{eq:linear_reservoir}, \eqref{eq:nonlinear_exosystem}, \eqref{eq:interconnection} with ${w(0)\in\mathcal K}$, the error \eqref{eq:nonlinear_manifold_error_dynamics} satisfies ${\dot{\varepsilon} = F\varepsilon}$, and $\mathcal M$ is exponentially attractive.
    \item\label{item:kkl_readout} There exists ${\Wout\in\R^{p\times n}}$ such that
    \begin{align}
        \Wout\pi(w)
        &=
        q(w),
        \qquad w\in\mathcal K,
        \label{eq:nonlinear_perfect_readout}
    \end{align}
    if and only if 
    ${\dim\operatorname{span}\begin{bsmallmatrix}\pi\\ q\end{bsmallmatrix}(\mathcal K)=\dim\operatorname{span}\pi(\mathcal K)}.$
    For ${w(0)\in\mathcal K}$, every such $\Wout$ gives a readout error \eqref{eq:nonlinear_readout_error} which converges exponentially to zero along the solutions of \eqref{eq:linear_reservoir},  \eqref{eq:interconnection}, \eqref{eq:nonlinear_exosystem}.
\end{enumerate}
}
\end{proposition}
\begin{proof}
(\ref{item:kkl_integral}) The argument follows \cite{andrieu_existence_2006}. Fix ${w\in\mathcal K}$. By assumption \A{asmp:nonlinear_backward}, ${\varphi_{-t}(w)\in\mathcal K\subseteq\mathcal O_\pi}$ for all ${t\geq0}$. Then \eqref{eq:kkl_invariance} implies
$${\frac{\dd}{\dd t}\bigl[e^{Ft}\pi(\varphi_{-t}(w))\bigr]=-e^{Ft}G\ell(\varphi_{-t}(w))}.$$ Since $\pi$ and $\ell$ are bounded on the compact set $\mathcal K$ and $F$ is Hurwitz by \A{asmp:linear_hurwitz}, integrating over $\Rp$ yields \eqref{eq:nonlinear_synchronization_map}.

(\ref{item:kkl_sync}) 
By assumption \A{asmp:nonlinear_backward}, ${w(0)\in\mathcal K}$ implies ${w(t)\in\mathcal K}$ for all ${t\geq0}$. Then \eqref{eq:kkl_invariance} gives 
$${\dot{\varepsilon} = Fx+G\ell(w)-\frac{\partial\pi}{\partial w}(w)s(w) = F\varepsilon},$$
which converges to zero exponentially by \A{asmp:linear_hurwitz}.

(\ref{item:kkl_readout}) By the equality of row and column rank, ${\dim\operatorname{span}\begin{bsmallmatrix}\pi\\ q\end{bsmallmatrix}(\mathcal K)}$ and ${\dim\operatorname{span}\pi(\mathcal K)}$ are the numbers of linearly independent components of $\begin{bsmallmatrix}\pi\\ q\end{bsmallmatrix}$ and of $\pi$, viewed as functions on $\mathcal K$. They coincide if and only if each component of $q$ is a linear combination of the components of $\pi$ on $\mathcal K$, that is, if and only if \eqref{eq:nonlinear_perfect_readout} holds for some $\Wout$, as in the proof of Proposition \ref{prop:luenberger}(\ref{item:exact_readout}). For every such $\Wout$, along a solution with ${w(0)\in\mathcal K}$, ${e=\Wout x-q(w)=\Wout\varepsilon}$, which converges exponentially to zero by (\ref{item:kkl_sync}).
\end{proof}
\noindent
Proposition \ref{prop:kkl_synchronization} strengthens Lemma \ref{lem:nonlinear_synchronization} for linear reservoirs: it provides an explicit expression of the synchronization map $\pi$, as the response of the reservoir to the past of the driving signal, and shows that $\mathcal M$ is globally attractive for exosystem states in $\mathcal K$. The integral \eqref{eq:nonlinear_synchronization_map} coincides with the expression of the synchronization map of linear reservoirs obtained in \cite{hart_generalised_2024} and, after a change of variables, with the integral defining \gls{kkl} observers \cite{andrieu_existence_2006}. Linear reservoirs driven by nonlinear exosystems may therefore be interpreted as functional \gls{kkl} observers of $q(w)$ whenever \eqref{eq:nonlinear_perfect_readout} holds, and as \gls{kkl} observers of $w$ if, in addition, $\pi$ is injective on $\mathcal K$ \cite{kazantzis_nonlinear_1997,andrieu_existence_2006}. The exact readout condition (\ref{item:kkl_readout}) states that each component of $q$ is, on $\mathcal K$, a linear combination of the components of $\pi$, and reduces to the rank condition of Proposition \ref{prop:luenberger} for ${\pi(w)=\Pi w}$, ${q(w)=Qw}$, and $\mathcal K$ spanning $\R^\nu$. The next result characterizes the signals which a linear reservoir with a linear readout can predict exactly.
\begin{corollary}[Immersion]
\label{cor:nonlinear_immersion}
Consider the reservoir \eqref{eq:linear_reservoir} driven by the exosystem \eqref{eq:nonlinear_exosystem} through the interconnection \eqref{eq:interconnection}, with ${p=m}$ and ${q=\ell}$, and let $\pi$ be the map of Lemma \ref{lem:nonlinear_synchronization}. Assume \A{asmp:linear_hurwitz}, \A{asmp:nonlinear_regularity}, \A{asmp:nonlinear_neutral_stability}, and \A{asmp:nonlinear_backward} hold, and that $W_{\mathrm{out}}$ satisfies \eqref{eq:nonlinear_perfect_readout}. Then, along every solution with ${w(0)\in\mathcal K}$, ${\xi=\pi(w)}$ satisfies
\begin{align}
    \dot{\xi}
    &=
    \bigl(F+GW_{\mathrm{out}}\bigr)\xi,
    &
    y
    &=
    W_{\mathrm{out}}\xi .
    \label{eq:nonlinear_immersed_system}
\end{align}
\end{corollary}
\begin{proof}
By \A{asmp:nonlinear_backward}, along a solution with ${w(0)\in\mathcal K}$, ${w(t)\in\mathcal K}$ for all ${t\geq0}$. Then \eqref{eq:kkl_invariance} and \eqref{eq:nonlinear_perfect_readout}, with ${q=\ell}$, imply ${\dot\xi=\frac{\partial\pi}{\partial w}(w)s(w) = F\pi(w)+G\ell(w) = (F+GW_{\mathrm{out}})\xi}$ and ${y=\ell(w)=W_{\mathrm{out}}\pi(w)=W_{\mathrm{out}}\xi}$.
\end{proof}
\noindent
Corollary \ref{cor:nonlinear_immersion} is the nonlinear counterpart of Proposition \ref{prop:free_modes}. The map $\pi$ sends every solution of the exosystem \eqref{eq:nonlinear_exosystem} in $\mathcal K$ to a solution of the closed-loop reservoir \eqref{eq:loop_closure}, so that the exosystem with output $\ell$ is \textit{immersed} into the linear system \eqref{eq:nonlinear_immersed_system} in the sense of \cite{isidori_nonlinear_1995}. Exact prediction with a linear reservoir and a linear readout is therefore possible only for signals which admit an $n$-dimensional autonomous linear realization, \textit{i.e.}, a finite-dimensional linear internal model. A related restriction appears in \cite{depersis_machines_2026}, where prediction is characterized through the modes visible at the output.

\begin{remark}[Nonlinear moment matching]\label{rem:moments}
The invariance equation \eqref{eq:nonlinear_invariance} is also the starting point of model reduction by moment matching for nonlinear systems \cite{astolfi_model_2010}. Under \A{asmp:nonlinear_regularity}--\A{asmp:nonlinear_neutral_stability}, which ensure that \eqref{eq:nonlinear_invariance} admits a locally unique solution $\pi$, the \emph{moment} of the reservoir \eqref{eq:general_reservoir}, \eqref{eq:linear_readout_definition} at the signal generator $(s,\ell)$ is the steady-state response ${w\mapsto\Wout\pi(w)}$ of the reservoir output to the signal generator \cite{astolfi_model_2010}. The exact readout condition \eqref{eq:nonlinear_perfect_readout} thus states that the moment of the trained reservoir at $(s,\ell)$ coincides with the target map $q$, and reservoir computing can be regarded as enforcing moment matching in a data-driven fashion \cite{scarciotti_data_2017}.
\end{remark}
\section{From analysis to reservoir design}\label{sec:design}
The interpretation of reservoir computing as a manifestation of the internal model principle has implications for reservoir analysis and design, which this section illustrates. On the analysis side, it explains limitations of exact training: the transverse modes left free by training decide the stability of the predictor, and the internal model embedded by training amplifies disturbances at frequencies defined by the spectrum of the signal generator. On the design side, it replaces empirical reservoir selection by an observer design problem whenever information on the origin of the signals is available.
\subsection{Fundamental limitations of linear reservoir computing}\label{sec:failure}
\begin{figure*}[t]
    \centering
{%
\begin{tikzpicture}
\begin{semilogxaxis}[
  width=0.325\textwidth, height=4.0cm, scale only axis=false,
  xlabel={$T$}, ylabel={$\Re\,\lambda_{\max}(F+GW_T)$},
  xmin=1, xmax=128, ymin=-3.2, ymax=2.6,
]
\addplot [rcGray, solid, thin] coordinates {(1,0) (128,0)};
\addplot [rcBlue, solid, mark=*, mark size=1.1pt] table [col sep=comma, x=T, y=rest] {data/example_abscissa.csv};
\addplot [rcPink, solid, mark=square*, mark size=1.1pt] table [col sep=comma, x=T, y=offset] {data/example_abscissa.csv};
\addplot [rcBlue, densely dotted, line width=0.9pt, mark=o, mark size=1.1pt] table [col sep=comma, x=T, y=transverse_rest] {data/example_abscissa.csv};
\addplot [rcPink, densely dotted, line width=0.9pt, mark=square, mark size=1.1pt] table [col sep=comma, x=T, y=transverse_offset] {data/example_abscissa.csv};
\end{semilogxaxis}
\end{tikzpicture}}\hfill
{%
\begin{tikzpicture}
\begin{axis}[
  width=0.325\textwidth, height=4.0cm, scale only axis=false,
  xlabel={$t$}, ylabel={$y$, $\hat y$},
  xmin=0, xmax=30, ymin=-3, ymax=3,
]
\addplot [rcTeal, solid, line width=1.3pt] table [col sep=comma, x=t, y=y] {data/example_prediction.csv};
\addplot [rcPink, dashed, line width=1.0pt] table [col sep=comma, x=t, y=stable] {data/example_prediction.csv};
\addplot [rcBlue, densely dotted, line width=0.9pt] table [col sep=comma, x=t, y=unstable] {data/example_prediction.csv};
\end{axis}
\end{tikzpicture}}\hfill
{%
\begin{tikzpicture}
\begin{axis}[
  width=0.325\textwidth, height=4.0cm, scale only axis=false,
  xlabel={$t$}, ylabel={$\hat y$},
  xmin=0, xmax=30, ymin=-1.1, ymax=1.1,
]
\addplot [rcPink, solid] table [col sep=comma, x=t, y=resonant] {data/example_disturbance.csv};
\addplot [rcBlue, dashed] table [col sep=comma, x=t, y=offresonant] {data/example_disturbance.csv};
\end{axis}
\end{tikzpicture}}
    \caption{Limitations of exact training in the example of Section \ref{sec:failure}, with ${\omega=1}$.
    (a) Largest real part of the eigenvalues of ${F+GW_T}$ (solid) and value of ${N^\top(F+GW_T)N}$ (dotted), for readouts trained by least squares on ${[0,T]}$ from ${x(0)=0}$ (circles) and ${x(0)=\Pi w(0)+6N}$ (squares).
    (b) Target signal (solid) and predictions of the exact readouts with transverse eigenvalue $-1$ (dashed) and $0.3$ (dotted), from ${\Pi w(0)}$ perturbed by $0.05$ along the transverse eigenvector.
    (c) Response of the predictor with transverse eigenvalue $-1$ to an input disturbance ${d(t)=0.01\sin(\Omega t)}$, for ${\Omega=\omega}$ (solid) and ${\Omega=2\omega}$ (dashed).}
    \label{fig:transverse}
\end{figure*}
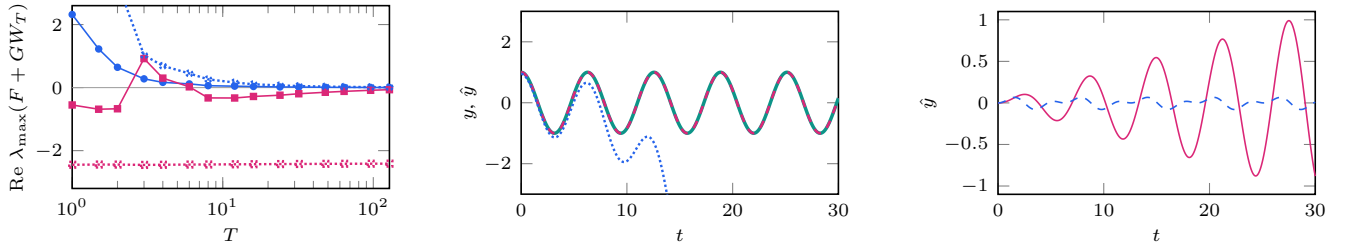
We consider a three-dimensional reservoir driven by a sinusoid to illustrate, in a simple setting, the mechanism by which exact training may produce unstable predictors and the fundamental limitation that the embedded internal model imposes on their robustness. 

Consider the reservoir \eqref{eq:linear_reservoir}, with ${n=3}$, ${m=1}$, ${F=-\operatorname{diag}(a_1,a_2,a_3)}$, with distinct ${a_i>0}$, and ${G=[\,1\ \ 1\ \ 1\,]^\top}$, driven through \eqref{eq:interconnection} by the exosystem \eqref{eq:linear_exosystem}, with ${\omega>0}$ and
\begin{align}
    S
    &=
    \begin{bmatrix}
        0 & \omega\\
        -\omega & 0
    \end{bmatrix},
    &
    L=Q
    &=
    \begin{bmatrix}
        1 & 0
    \end{bmatrix}.
    \label{eq:example_exosystem}
\end{align}
Then \A{asmp:linear_hurwitz}--\A{asmp:linear_observability} hold, and the solution $\Pi$ of \eqref{eq:sylvester}, whose $i$th row is ${(a_i,-\omega)/(a_i^2+\omega^2)}$, has rank two. With $\Pi^\dagger$ and the unit vector $N$ spanning $\ker\Pi^\top$ as in Remark \ref{rem:transverse}, every solution of \eqref{eq:perfect_readout} satisfies ${\Wout=\Wout(\Pi\Pi^\dagger+NN^\top)}$, so every exact readout may be expressed as
\begin{align}
    \Wout
    &=
    Q\Pi^\dagger+\kappa N^\top,
    \qquad
    \kappa\in\R.
    \label{eq:example_readouts}
\end{align}
By Proposition \ref{prop:free_modes} and Remark \ref{rem:transverse}, the closed loop \eqref{eq:loop_closure} has the eigenvalues ${\lambda_{\pm} = \pm i\omega}$ and the remaining transverse eigenvalue ${N^\top(F+G\Wout)N}$, which, since ${\Pi^\dagger N=0}$, reads
\begin{align}
    \lambda_{\mathrm{tr}}
    &=
    N^\top FN+\kappa N^\top G.
    \label{eq:example_transverse}
\end{align}
Since ${N^\top G\neq0}$ by Remark \ref{rem:transverse}, every real value of $\lambda_{\mathrm{tr}}$ may be attained by an exact readout. For ${\omega=1}$ and ${(a_1,a_2,a_3)=(1,2,3)}$, ${N=\tfrac{1}{\sqrt{51}}[\,1\ \ -5\ \ 5\,]^\top}$, so 
$${\lambda_{\mathrm{tr}}=-\tfrac{42}{17}+\tfrac{\kappa}{\sqrt{51}}},$$ 
and the closed loop is unstable for ${\kappa>42\sqrt{51}/17\approx17.6}$ despite zero training error, see Fig. \ref{fig:transverse}(b). 

The value of $\kappa$ is selected by the training transient. Indeed, since ${N^\top\Pi=0}$, the component of the reservoir state along $N$ is ${N^\top x(t)=N^\top e^{Ft}\bigl(x(0)-\Pi w(0)\bigr)}$, which vanishes on synchronized data, ${x(0)=\Pi w(0)}$. In that case, the least-squares problem \eqref{eq:least_squares_minimization} does not depend on $\kappa$. Otherwise, $\kappa$ is determined by the transient. For example, let ${w(0)=(1,0)}$ and let $W_T$ denote the minimizer of \eqref{eq:least_squares_minimization} on $[0,T]$. As $T$ grows, $W_T$ approaches an exact readout \eqref{eq:example_readouts} with transverse eigenvalue ${\lambda_{\mathrm{tr}}=0}$ when ${x(0)=0}$, and ${\lambda_{\mathrm{tr}}\approx-2.41}$ when ${x(0)=\Pi w(0)+6N}$, see Fig. \ref{fig:transverse}(a). For finite $T$, however, $W_T$ does not satisfy \eqref{eq:perfect_readout}, and the eigenvalues of ${F+GW_T}$ near ${\pm i\omega}$ leave the imaginary axis: ${F+GW_T}$ is unstable for every $T$ in Fig. \ref{fig:transverse}(a) when ${x(0)=0}$, and Hurwitz for ${T\geq8}$ when ${x(0)=\Pi w(0)+6N}$, so that the prediction grows or decays, respectively.

The same example illustrates a fundamental limitation of exact prediction.  Let ${v_\pm}$ be eigenvectors of $S$ associated with ${\lambda_{\pm} = \pm i\omega}$. By \eqref{eq:closed_loop_intertwining}, ${\Pi v_\pm}$ are eigenvectors of ${F+G\Wout}$ associated with ${\lambda_{\pm}}$. These modes are controllable from $G$, since ${(F+G\Wout,G)}$ is controllable by \A{asmp:linear_controllability}, and observable from $\Wout$, since ${\Wout\Pi v_\pm=Lv_\pm\neq0}$ by \A{asmp:linear_observability}. Hence ${\lambda_{\pm}}$ are poles of the transfer function ${\Wout(sI_n-F-G\Wout)^{-1}G}$ from a disturbance $d(t)\in\R^m$, acting via ${u=\Wout x+d}$, to the prediction $\hat y$. The predictor thus has infinite gain at the frequency of the signal it reproduces: for ${\lambda_{\mathrm{tr}}<0}$, the response to ${d(t)=\sin(\omega t)}$ is unbounded, whereas the response to ${d(t)=\sin(2\omega t)}$ remains bounded, see Fig. \ref{fig:transverse}(c).
\subsection{Reservoir design as observer design}\label{sec:observer_design}
\begin{figure*}[t]
    \centering
{%
\begin{tikzpicture}
\begin{axis}[
  width=0.325\textwidth, height=4.0cm, scale only axis=false,
  xlabel={$t$}, ylabel={$\theta$, $x_1$},
  xmin=0, xmax=10, ymin=-3.2, ymax=3.2,
]
\addplot [rcGray, solid, line width=1.4pt] table [col sep=comma, x=t, y=theta] {data/pendulum_transient.csv};
\addplot [rcBlue, solid] table [col sep=comma, x=t, y=x1_a05] {data/pendulum_transient.csv};
\addplot [rcPink, dashed] table [col sep=comma, x=t, y=x1_a1] {data/pendulum_transient.csv};
\addplot [rcNavy, dashdotted] table [col sep=comma, x=t, y=x1_a2] {data/pendulum_transient.csv};
\draw [rcGray, thin] (axis cs:0,1.2) rectangle (axis cs:1.5,3.1);
\coordinate (magnified region) at (axis cs:1.5,3.1);
\coordinate (inset position) at (rel axis cs:0.34,0.55);
\end{axis}
\begin{axis}[
  name=inset, at={(inset position)}, anchor=south west, width=1.6cm, height=1.05cm, scale only axis=true,
  xmin=0, xmax=1.5, ymin=1.2, ymax=3.1, xtick={0,1.5}, ytick={1.5,3}, xticklabels={$0$,$1.5$}, yticklabels={$1.5$,$3$},
  tick label style={font=\tiny}, axis background/.style={fill=white}, axis line style={rcGray, thin},
]
\addplot [rcGray, solid, line width=1.2pt] table [col sep=comma, x=t, y=theta] {data/pendulum_transient.csv};
\addplot [rcBlue, solid] table [col sep=comma, x=t, y=x1_a05] {data/pendulum_transient.csv};
\addplot [rcPink, dashed] table [col sep=comma, x=t, y=x1_a1] {data/pendulum_transient.csv};
\addplot [rcNavy, dashdotted] table [col sep=comma, x=t, y=x1_a2] {data/pendulum_transient.csv};
\end{axis}
\draw [rcGray, thin] (magnified region) -- (inset.south west);
\end{tikzpicture}
}\hfill
{%
\begin{tikzpicture}
\begin{semilogyaxis}[
  width=0.325\textwidth, height=4.0cm, scale only axis=false,
  xlabel={$t$}, ylabel={$\|x-w\|/\|w\|$},
  xmin=0, xmax=20, ymin=1e-9, ymax=10,
]
\addplot [rcBlue, solid] table [col sep=comma, x=t, y=a05] {data/pendulum_sync_error.csv};
\addplot [rcPink, dashed] table [col sep=comma, x=t, y=a1] {data/pendulum_sync_error.csv};
\addplot [rcNavy, dashdotted] table [col sep=comma, x=t, y=a2] {data/pendulum_sync_error.csv};
\end{semilogyaxis}
\end{tikzpicture}}\hfill
{%
\begin{tikzpicture}
\begin{loglogaxis}[
  width=0.325\textwidth, height=4.0cm, scale only axis=false,
  xlabel={$a$}, ylabel={$e_{\mathrm{rms}}$},
  xmin=0.25, xmax=10, ymin=1.8e-4, ymax=9e-4,
  xtick={0.25,0.5,1,2,5,10}, xticklabels={$0.25$,$0.5$,$1$,$2$,$5$,$10$},
  ytick={2e-4,3e-4,5e-4,8e-4}, yticklabels={$2\cdot10^{-4}$,$3\cdot10^{-4}$,$5\cdot10^{-4}$,$8\cdot10^{-4}$},
]
\addplot [name path=upper, draw=none] table [col sep=comma, x=a, y=upper] {data/pendulum_noise.csv};
\addplot [name path=lower, draw=none] table [col sep=comma, x=a, y=lower] {data/pendulum_noise.csv};
\addplot [rcBlue!25] fill between [of=upper and lower];
\addplot [rcBlue, solid, mark=*, mark size=1.2pt] table [col sep=comma, x=a, y=mean] {data/pendulum_noise.csv};
\end{loglogaxis}
\end{tikzpicture}}
    \caption{Reservoir design for the frictionless pendulum \eqref{eq:pendulum},
    released from rest at ${w_1=2.5}$.
    (a) Target ${y=w_1}$ (thick) and estimate ${\hat y=x_1}$ of the reservoir \eqref{eq:pendulum_reservoir}, driven by \eqref{eq:pendulum} and starting at rest, for the eigenvalues
    $-a$ and $-2a$ of \eqref{eq:pendulum_error}, with ${a=0.5}$ (solid), ${a=1}$ (dashed),
    and ${a=2}$ (dash-dotted).
    (b) Relative synchronization error ${\|x-w\|/\|w\|}$ for the same gains and line styles.
    (c) Relative \gls{rms} estimation error \eqref{eq:relative_rms_error} under
    measurement noise of standard deviation $10^{-2}$, as a function of $a$
    (mean over twenty noise realizations, with the band of one standard deviation shaded).}
    \label{fig:pendulum}
\end{figure*}
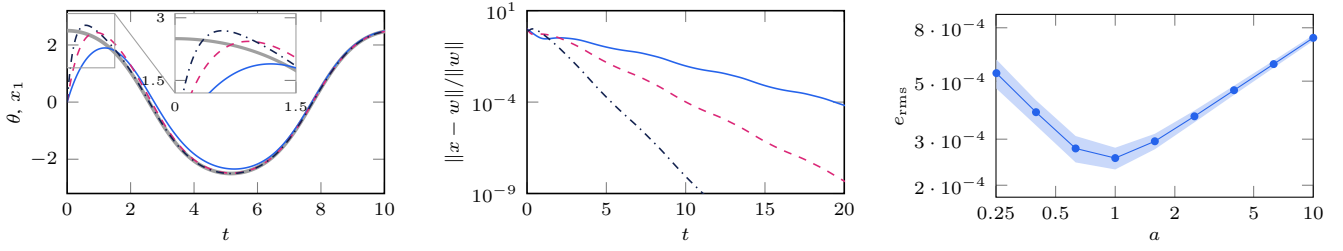
Reservoir computing practice often selects large reservoirs at random, relying on their dimension and nonlinearity for expressive power \cite{zhang_survey_2023}, an approach whose costs in dimension, tuning, and interpretability are well-documented \cite{lukosevicius_reservoir_2009,gallicchio_deep_2017}. When information about the origin of the signal is available, the observer interpretation of Sections \ref{sec:linear} and \ref{sec:nonlinear} suggests instead a parsimonious design: a reservoir which contains a copy of the signal generator, corrected by the measured signal, as in classical  observer design \cite{luenberger_introduction_1971,khalil_nonlinear_2002}. We illustrate this point on the frictionless pendulum
\begin{align}
    \dot w_1
    &=
    w_2,
    &
    \dot w_2
    &=
    -\sin w_1,
    &
    v
    &=
    w_1,
    \label{eq:pendulum}
\end{align}
that is, the nonlinear exosystem \eqref{eq:nonlinear_exosystem}, with ${w(t)\in\R^2}$, ${s(w)=(w_2,-\sin w_1)}$, ${\ell(w)=w_1}$, and target map ${q=\ell}$. Consider the two-dimensional reservoir
\begin{align}
    \dot x_1 &= x_2+k_1(u-x_1), &
    \dot x_2 &= -\sin u+k_2(u-x_1),
    \label{eq:pendulum_reservoir}
\end{align}
with gains ${k_1,k_2>0}$, driven by \eqref{eq:pendulum} through the interconnection \eqref{eq:interconnection}. A direct computation shows that the identity map ${\pi(w)=w}$ solves the invariance equation \eqref{eq:nonlinear_invariance} on $\R^2$, and that the error ${\varepsilon=x-w}$, defined as in \eqref{eq:nonlinear_manifold_error_dynamics}, satisfies
\begin{align}
    \dot\varepsilon
    &=
    \begin{bmatrix}
        -k_1 & 1\\
        -k_2 & 0
    \end{bmatrix}
    \varepsilon .
    \label{eq:pendulum_error}
\end{align}
The eigenvalues of \eqref{eq:pendulum_error} are the roots of the characteristic polynomial ${\psi(\lambda)=\lambda^2+k_1\lambda+k_2}$, and can be assigned to any conjugate pair in the open left half-plane through ${k_1,k_2>0}$, using linear observer design tools \cite{luenberger_introduction_1971}. For example, choosing ${\psi(\lambda)=(\lambda+a)(\lambda+2a)}$, with ${a>0}$, gives ${k_1=3a}$ and ${k_2=2a^2}$.
A direct computation also shows that ${\Wout=[\,1\ \ 0\,]}$ satisfies the exact readout condition \eqref{eq:nonlinear_perfect_readout}. This readout is the unique minimizer of \eqref{eq:least_squares_minimization} on synchronized data, ${x=w}$, whenever ${\int_0^Tww^\top\dd t\succ0}$ \cite{padoan_geometric_2017}, which holds on every non-equilibrium solution of \eqref{eq:pendulum}. Finally, closing the loop with ${u=\Wout x=x_1}$ turns \eqref{eq:pendulum_reservoir} into the pendulum \eqref{eq:pendulum} itself. The reservoir thus contains an internal model of the signal generator, and its synchronization rate is a design variable, independent of training.

Fig. \ref{fig:pendulum} shows the target ${y=w_1}$, the estimate ${\hat y=x_1}$ of the driven reservoir, and the relative synchronization error for ${a\in\{0.5,1,2\}}$, along the solution released from rest at ${w_1(0)=2.5}$. The ease of this design allows one to trade synchronization speed against noise sensitivity. With additive measurement noise, ${u=w_1+d}$, a first-order expansion of \eqref{eq:pendulum_reservoir} gives
\begin{align*}
    \dot\varepsilon_1
    &=
    \varepsilon_2-k_1\varepsilon_1+k_1d,
    &
    \dot\varepsilon_2
    &=
    -k_2\varepsilon_1+(k_2-\cos w_1)d,
\end{align*}
which shows that large gains accelerate synchronization but amplify the noise entering through the output-injection terms, whereas small gains slow down synchronization. Fig. \ref{fig:pendulum} shows, as a function of $a$, the relative \gls{rms} estimation error
\begin{align}
    e_{\mathrm{rms}}=\frac{\|\hat y-y\|_{L^2}}{\|y\|_{L^2}},
    \label{eq:relative_rms_error}
\end{align}
where $\|\cdot\|_{L^2}$ denotes the $L^2$ norm on ${[2T_0,4T_0]}$ and $T_0$ is the period of the solution. The error is minimized at an intermediate gain, a trade-off which the observer viewpoint makes explicit.
\section{Conclusion}\label{sec:conclusion}
The paper has shown that learning in reservoir computing is a manifestation of the internal model principle, establishing connections with output regulation, observer design, and moment matching both for linear and nonlinear systems.  Several directions remain open. Our perspective suggests the study of fundamental limitations of reservoir computing, expressed for linear reservoirs by sensitivity integrals, of the interplay between the cost minimized by training and the closed-loop dynamics of the predictor, and of the role of persistence of excitation \cite{padoan_geometric_2017} and least-squares moment matching \cite{padoan_least_2023} in readout training. Principled reservoir design in physical substrates is another natural next step \cite{tanaka_recent_2019,stepney_physical_2024}.
\section*{Acknowledgment} 
 We gratefully acknowledge the authors of \cite{bizyaeva_kkl_2026}, whose closely related work was developed independently and in parallel with ours. We thank them for valuable discussions and for bringing their work to our attention. 
\bibliographystyle{ieeetr}
\bibliography{ref}
\end{document}